\documentclass[12pt]{article}
\usepackage{amsmath}
\usepackage{graphicx,psfrag,epsf}
\usepackage{enumerate}
\usepackage{natbib}

\usepackage{amssymb}
\usepackage{amsfonts}
\usepackage{multirow}
\usepackage{amsthm}

\usepackage{bm}
\usepackage{color}
\usepackage{floatrow}
\usepackage{units}
\usepackage{url}
\usepackage{verbatim}
\usepackage{slashbox}
\usepackage{epstopdf}
\usepackage{booktabs,caption,fixltx2e}
\usepackage[flushleft]{threeparttable}
\usepackage{rotating}
\usepackage{multirow}
\usepackage{physics}
\usepackage{nameref}
\usepackage{placeins} % gives FloatBarrier to fix position of Figures

\newtheorem{theorem}{Theorem}
\newtheorem{lemma}{Lemma}

\theoremstyle{definition}

\newcommand{\mbf}[1]{\mathbf{#1}}

\newcommand{\ind}{\stackrel{\mathrm{ind}}{\sim}}
\newcommand{\blind}{1}

\begin{document}

\def\spacingset#1{\renewcommand{\baselinestretch}%
{#1}\small\normalsize} \spacingset{1}

%%%%%%%%%%%%%%%%%%%%%%%%%%%%%%%%%%%%%%%%%%%%%%%%%%%%%%%%%%%%%%%%%%%%%%%%%%%%%%

\if1\blind
{
  \title{\bf Bayesian Estimation Under Informative Sampling with Unattenuated Dependence}
  \author{Matthew R. Williams\hspace{.2cm}\\
    Substance Abuse and Mental Health Services Administration\\
    and \\
    Terrance D. Savitsky \\
    U.S. Bureau of Labor Statistics}
  \maketitle
} \fi

\if0\blind
{
  \bigskip
  \bigskip
  \bigskip
  \begin{center}
    {\LARGE\bf Bayesian Estimation Under Informative Sampling with Unattenuated Dependence}
\end{center}
  \medskip
} \fi

\bigskip
\begin{abstract}
An informative sampling design leads to unit inclusion probabilities that are correlated with the response variable of interest. However, multistage sampling designs may also induce higher order dependencies, which are typically ignored in the literature when establishing consistency of estimators for survey data under a condition requiring asymptotic independence among the unit inclusion probabilities. We refine and relax this condition of asymptotic independence or asymptotic factorization and demonstrate that consistency is still achieved in the presence of residual sampling dependence. A popular approach for conducting inference on a population based on a survey sample is the use of a pseudo-posterior, which uses sampling weights based on first order inclusion probabilities to exponentiate the likelihood. We show that the pseudo-posterior is consistent not only for survey designs which have asymptotic factorization, but also for designs with residual or unattenuated dependence. Using the complex sampling design of the National Survey on Drug Use and Health, we explore the impact of multistage designs and order based sampling. The use of the survey-weighted pseudo-posterior together with our relaxed requirements for the survey design establish a broad class of analysis models that can be applied to a wide variety of survey data sets.
\end{abstract}

\noindent%
\vfill

\newpage
\spacingset{1.45} % DON'T change the spacing!

\section{Introduction}
Bayesian formulations are increasingly popular for modeling hypothesized distributions with complicated dependence structures.
The primary interest of the data analyst is to perform inference about a finite population generated from an unknown model, $P_{0}$.
The observed data are often collected from a sample taken from that finite population under a complex sampling design distribution, $P_{\nu}$, resulting in probabilities of inclusion that are associated with the variable of interest. This association could result in an observed data set consisting of units that are not independent and identically distributed. This association induces a correlation between the response variable of interest and the inclusion probabilities.
Sampling designs that induce this correlation are termed, ``informative", and the balance of information in the sample is different from that in the population.
Failure to account for this dependence caused by the sampling design could bias estimation of parameters that index the joint distribution hypothesized to have generated the population \citep{holt:1980}. While emphasis is often placed on the first order inclusions probabilities (individual probabilities of selection), an ``informative'' design may also have other features such as clustering and stratification that use population information and impact higher order joint inclusion probabilities. The impact of these higher order terms are more subtle but we will demonstrate that they can also impact bias and consistency.

\citet{2015arXiv150707050S} proposed an automated approach that formulates a sampling-weighted pseudo-posterior density by exponentiating each likelihood contribution by a sampling weight constructed to be inversely proportional to its marginal inclusion probability, $\pi_{i} = P\left(\delta_{i} = 1\right)$, for units, $i = 1,\ldots,n$, where $n$ denotes the number of units in the observed sample.  The inclusion of unit, $i$, from the population, $U$, in the sample is indexed by $\delta_{i} \in \{0,1\}$.  They restrict the class of sampling designs to those where the pairwise dependencies among units attenuate to $0$ in the limit of the population size, $N$, (at order $N$) to guarantee posterior consistency of the pseudo-posterior distribution estimated on the sample data, at $P_{0}$ (in $L_{1}$). While some sampling designs will meet this criterion, many won't; for example, a two-stage clustered sampling design where the number of clusters increases with $N$, but the number of units in each cluster remain relatively fixed such that the dependence induced at the second stage of sampling never attenuates to $0$.  A common example are designs which select households as clusters. Despite the lack of theoretical results demonstrating consistency of estimators under this scenario, use of first order sampling weights performs well in practice. This work provides new theoretical conditions for when consistency can be achieved and provides some examples for when these conditions are violated and consistency may not be achieved.

\subsection{Motivating Example: The National Survey on Drug Use and Health}
Our motivating survey design is the National Survey on Drug Use and Health (NSDUH), sponsored by the Substance Abuse and Mental Health Services Administration (SAMHSA). NSDUH is the primary source for statistical information on illicit drug use, alcohol use, substance use disorders (SUDs), mental health issues, and their co-occurrence for the civilian, non institutionalized population of the United States.  The NSDUH employs a multistage state-based design
\citep{MRB:Sampling:2014},
with the earlier stages defined by geography within each state in order to select households (and group quarters) nested within these geographically-defined primary sampling units (PSUs). The sampling frame was stratified implicitly by sorting the first-stage sampling units by a core-based statistical area (CBSA) and socioeconomic status indicator and by the percentage of the population that is non-Hispanic and white. First stage units (census tracts) were then selected with probability proportionate to a composite size measure based on age groups. This selection was performed `systematically' along the sort order gradient. Second and third stage units (census block groups and census blocks) were sorted geographically and selected with probability proportionate to size (PPS) sequentially along the sort order. Fourth stage dwelling units (DU) were selected systematically with equal probability, selecting every $k^{th}$ DU after a random starting point. Within households, 0, 1 or 2 individuals were selected with unequal probabilities depending on age with youth (age 12- 17) and young adults (age 18-25) over-sampled.

This paper provides conditions for asymptotic consistency for designs like the NSDUH, which are characterized by:
\begin{itemize}
\item Cluster sampling, such as selecting only one unit per cluster, or selecting multiple individuals from a dwelling unit.
\item Population information used to sort sampling units along gradients.
\end{itemize}
Both features are common, in practice, and create pairwise sampling dependencies that do not attenuate even if the population grows.  The consistency of estimators under these sampling designs are not addressed in the literature. For example, we will examine the relationship between depression and smoking. 
%(ADD citationation relationship to age/ pop density - MHDT's). 
Cigarette use and depression vary by age, metropolitan vs. non-metropolitan status, education level, and other demographics \citep{DT:2014,MHDT:2014}.
Both smoking and depression have the potential to cluster geographically and within dwelling units, since these related demographics may cluster. Yet the current literature, such as in \citet{2015arXiv150707050S}, is silent on the issue of non-ignorable clustering that may be informative (i.e. related to the response of interest). The results presented in this work establish conditions for a wide variety of survey designs and provide a theoretical justification that this relationship can be estimated consistently even under a complex multistage design such as the NSDUH.

\subsection{Review of Methods to Account for Dependent Sampling}
For consistency results, assumptions of approximate or asymptotic independence of sample selection (or factorization of joint inclusion probabilities into a product of individual inclusion probabilities) are ubiquitous. For example, \citet{Isaki82} assume asymptotic factorization to demonstrate the consistency of the Horvitz-Thompson estimator and related regression estimators. More recently, \citet{zbMATH06017974} used a similar assumption to demonstrate consistency of survey-weighted regression trees and \citet{2015arXiv150707050S} used it to show consistency of a survey-weighted pseudo-posterior.

\citet[Ch.2]{chambers2003analysis} review the construction of a sample likelihood using a Bayes rule expression for the population likelihood defined on the sampled units,  $f_{s}(y) = f_{U}(y|I = 1)$ (similar to \citet{pkr:1998}).
They explicitly state the assumption that the ``sample inclusion for any particular population unit is independent of that for any other unit (and is determined by the outcome of a zero-one random variable, $I$)''.
The further assumption of independence of the population units stated in \citet{pkr:1998} means that weighting each likelihood contribution multiplied together in the sample is an approximation of the likelihood for the $N$ population units.

\citet{pkr:1998} maintains the assumption of unconditional independence of the population units, but defines two classes of sampling designs: (1) The first class is independent, with replacement sampling, so the sample inclusions are all independent.
(2) The second class is some selected with replacement designs that are asymptotically independent.
\citet{chambers2003analysis} discuss the pseudo-likelihood (and cite \citet{Kish74}, \citet{Binder83} and \citet{Godambe86})  for estimation via a weighted score function. They assume that the correlation between inclusion indicators has an \emph{expected value} of 0, where the expectation is with respect to the population generating distribution. We note that they do not assume this correlation to be exactly equal to 0. However this condition still appears to be more restrictive than that of asymptotic factorization in which deviations from factorization shrink to 0 at a rate inverse to the population size $N$: $\order{N^{-1}}$.

The assumptions above are relied on to show consistency. However in practice, approximate sampling independence is only assumed for the first stage or primary sampling units (PSUs), with dependence between secondary units within these clusters commonly assumed. This setup is the defacto approach for design-based variance estimation (for example, see \citet[Ch.3]{heeringa2010applied} and  \cite{Rao92}) and is used in all the major software packages for analyzing survey data. One goal of the current work is to reconile this discrepancy by extending the class of designs for which consistency results are available to cover designs seen in practice such as those for which design-based variance estimation strategies already exist.

We focus on extending the results of the survey-weighted pseudo-posterior method of \citet{2015arXiv150707050S} which provides for flexible modeling of a very wide class of population generating models. By refining and relaxing the conditions on factorization, we expand results to include many common sampling designs. These conditions for the sampling designs can be applied to generalize many of the other consistency results mentioned above. There are some population models of interest for which marginal inclusion probabilities may not be sufficient and pairwise inclusion probabilities and composite likelihoods can be used to achieve consistent results \citep{2016yi, 2017pair}. However, \citet{2017pair} demonstrate that both a very specific population model (for example conditional behavior of spouse-spouse pairs within households) and specific sample design (differential selection of pairs of individuals within a household related to outcome) are needed for marginal weights to lead to bias. In the usual setting of inference on a population of individuals (rather than on a population of joint relationships \emph{within} households), pairwise weights and marginal weights are numerically similar, converging to one another for moderate sample sizes. The theory presented in the current work also clarifies why both approaches lead to consistent results. Furthermore, the current work also applies when individual units are mutually exclusive; for example, only selecting one individual from a household to the exclusion of all others. Such designs are not covered by the composite likelihood with pairwise weights approach, which require non-zero joint inclusion probabilities.

The remainder of this work proceeds as follows: In section \ref{sec:pseudop} we briefly review the pseudo-posterior approach to account for informative sampling via the exponentiaion of the likelihood with sampling weights. Our main result, presented in section \ref{results}, provides the formal conditions controlling for sampling dependence. In section \ref{sec:sims}, we provide two simulations. We first demonstrate consistency for a multistage survey design analogous to the NSDUH. We next create a pathological design based on sorting. The design violates our assumptions for sampling dependence and estimates fail to converge. However, we show that this design will lead to consistency if embedded within stratified or clustered designs. Lastly, we revisit the NSDUH with a simple example (section \ref{sec:NSDUH}) and provide some conclusions (section \ref{sec:conc}).

\section{Pseudo-Posterior Estimator to Account for Informative Sampling}
\label{sec:pseudop}
We briefly review the pseudo-likelihood and associated pseudo-posterior as constructed in \citet{2015arXiv150707050S} and revisited by \citet{2017pair}.

Suppose there exists a Lebesgue measurable population-generating density, $\pi\left(y\vert\bm{\lambda}\right)$, indexed by parameters,
$\bm{\lambda} \in \Lambda$. Let $\delta_{i} \in \{0,1\}$ denote the sample inclusion indicator for units $i = 1,\ldots,N$ from the population.  The density for the observed sample is denoted by, $\pi\left(y_{o}\vert\bm{\lambda}\right) = \pi\left(y\vert \delta = 1,\bm{\lambda}\right)$, where ``$o$" indicates ``observed".

The plug-in estimator for posterior density under the analyst-specified model for $\bm{\lambda} \in \Lambda$ is
\begin{equation}
\hat{\pi}\left(\bm{\lambda}\vert \mathbf{y}_{o},\tilde{\mathbf{w}}\right) \propto \left[\mathop{\prod}_{i = 1}^{n}p\left(y_{o,i}\vert \bm{\lambda}\right)^{\tilde{w}_{i}}\right]\pi\left(\bm{\lambda}\right), \label{pseudolike}
\end{equation}
where
$\mathop{\prod}_{i=1}^{n}p\left(y_{o,i}\vert \bm{\lambda}\right)^{\tilde{w}_{i}}$ denotes the pseudo-likelihood for observed sample responses, $\mathbf{y}_{o}$. The joint prior density on model space assigned by the analyst is denoted by $\pi\left(\bm{\lambda}\right)$.  The sampling weights, $\{\tilde{w}_{i} \propto 1/\pi_{i}\}$, are inversely proportional to unit inclusion probabilities and normalized to sum to the sample size, $n$. Let $\hat{\pi}$ denote the noisy approximation to posterior distribution, $\pi$, based on the data, $\mathbf{y}_{o}$, and sampling weights, $\{\tilde{\mathbf{w}}\}$, confined to those units \emph{included} in the sample, $S$.

\section{Consistency of the Pseudo Posterior Estimator}
\label{results}
A sampling design is defined by placing a \emph{known} distribution on a vector of inclusion indicators, $\bm{\delta}_{\nu} = \left(\delta_{\nu 1},\ldots,\delta_{\nu N_{\nu}}\right)$, linked to the units comprising the population, $U_{\nu}$.  The sampling distribution is subsequently used to take an \emph{observed} random sample of size $n_{\nu} \leq N_{\nu}$.
Our conditions needed for the main result employ known marginal unit inclusion probabilities, $\pi_{\nu i} = \mbox{Pr}\{\delta_{\nu i} = 1\}$ for all $i \in U_{\nu}$ and the second order pairwise probabilities, $\pi_{\nu ij} = \mbox{Pr}\{\delta_{\nu i} = 1 \cap \delta_{\nu j} = 1\}$ for $i,j \in U_{\nu}$, which are obtained from the joint distribution over $\left(\delta_{\nu 1},\ldots,\delta_{\nu N_{\nu}}\right)$. We denote the sampling distribution by $P_{\nu}$.

Under informative sampling, the inclusion probabilities
are formulated to depend on the finite population data values, $\mathbf{X}_{N_{\nu}} = \left(\mbf{X}_{1},\ldots,\mbf{X}_{N_{\nu}}\right)$. Information from the population is used to determine size measures for unequal selection  $\pi_{\nu i}$ and used to establish clustering and stratification which determine joint inclusions probabilities $\pi_{\nu ij}$.
Since the balance of information is different between the population and a resulting sample, a posterior distribution for $\left(\mbf{X}_{1}\delta_{\nu 1},\ldots,\mbf{X}_{N_{\nu}}\delta_{\nu N_{\nu}}\right)$ that ignores the distribution for $\bm{\delta}_{\nu}$ will not lead to consistent estimation.

Our task is to perform inference about the population generating distribution, $P_{0}$, using the observed data taken under an informative sampling design.  We account for informative sampling by ``undoing" the sampling design with the weighted estimator,
\begin{equation}
p^{\pi}\left(\mbf{X}_{i}\delta_{\nu i}\right) := p\left(\mbf{X}_{i}\right)^{\delta_{\nu i}/\pi_{\nu i}},~i \in U_{\nu},
\end{equation}
which weights each density contribution, $p(\mbf{X}_{i})$, by the inverse of its marginal inclusion probability. This approximation for the population likelihood produces the associated pseudo-posterior,
\begin{equation}\label{inform_post}
\Pi^{\pi}\left(B\vert \mbf{X}_{1}\delta_{\nu 1},\ldots,\mbf{X}_{N_{\nu}}\delta_{\nu N_{\nu}}\right) = \frac{\mathop{\int}_{P \in B}\mathop{\prod}_{i=1}^{N_{\nu}}\frac{p^{\pi}}{p_{0}^{\pi}}(\mbf{X}_{i}\delta_{\nu i})d\Pi(P)}{\mathop{\int}_{P \in \mathcal{P}}\mathop{\prod}_{i=1}^{N_{\nu}}\frac{p^{\pi}}{p_{0}^{\pi}}(\mbf{X}_{i}\delta_{\nu i})d\Pi(P)},
\end{equation}
that we use to achieve our required conditions for the rate of contraction of the pseudo-posterior distribution on $P_{0}$.  We note that both $P$ and $\bm{\delta}_{\nu}$ are random variables defined on the space of measures ($\mathcal{P}$ and $ B \subseteq \mathcal{P}$) and the distribution, $P_{\nu}$, governing all possible samples, respectively.  An important condition on $P_{\nu}$ formulated in \citet{2015arXiv150707050S} that guarantees contraction of the pseudo-posterior on $P_{0}$ restricts pairwise inclusion dependencies to asymptotically attenuate to $0$.  This restriction narrows the class of sampling designs for which consistency of a pseudo-posterior based on marginal inclusion probabilities may be achieved.  We will replace their condition that requires marginal factorization of all pairwise inclusion probabilities with a less restrictive condition allowing for non-factorization for a small partition of pairwise inclusion probabilities. This expands the allowable class of sampling designs under which frequentist consistency may be guaranteed.  We assume measurability for the sets on which we compute prior, posterior and pseudo-posterior probabilities on the joint product space, $\mathcal{X}\times\mathcal{P}$.  For brevity, we use the superscript, $\pi$, to denote the dependence on the known sampling probabilities, $\{\pi_{\nu ij}\}_{i,j \in U_{\nu}}$; for example,
\[
\displaystyle\Pi^{\pi}\left(B\middle\vert \mbf{X}_{1}\delta_{\nu 1},\ldots,\mbf{X}_{N_{\nu}}\delta_{\nu N_{\nu}}\right) := \Pi\left(B\middle\vert \left(\mbf{X}_{1}\delta_{\nu 1},\ldots,\mbf{X}_{N_{\nu}}\delta_{\nu N_{\nu}}\right),
	\{\pi_{\nu ij}: i,j \in U_{\nu} \} \right).
\]

Our main result is achieved in the limit as $\nu\uparrow\infty$, under the countable set of successively larger-sized populations, $\{U_{\nu}\}_{\nu \in \mathbb{Z}^{+}}$. 
% We define the associated rate of convergence notation, $\order{b_{\nu}}$, to denote $\mathop{\lim}_{\nu\uparrow\infty}\frac{\order{b_{\nu}}}{b_{\nu}} = 0$.
We define the associated rate of convergence notation, $a_{\nu} = \order{b_{\nu}}$, to denote $ |a_{\nu}| \le M |b_{\nu}|$ for a constant $M > 0$.

\subsection{Empirical process functionals}\label{empirical}
We employ the empirical distribution approximation for the joint distribution over population generation and the draw of an informative sample that produces our observed data to formulate our results.  Our empirical distribution construction follows \citet{breslow:2007} and incorporates inverse inclusion probability weights, $\{1/\pi_{\nu i}\}_{i=1,\ldots,N_{\nu}}$, to account for the informative sampling design,
\begin{equation}
\mathbb{P}^{\pi}_{N_{\nu}} = \frac{1}{N_{v}}\mathop{\sum}_{i=1}^{N_{\nu}}\frac{\delta_{\nu i}}{\pi_{\nu i}}\delta\left(\mbf{X}_{i}\right),
\end{equation}
where $\delta\left(\mbf{X}_{i}\right)$ denotes the Dirac delta function, with probability mass $1$ on $\mbf{X}_{i}$ and we recall that $N_{\nu} = \vert U_{\nu} \vert$ denotes the size of of the finite population. This construction contrasts with the usual empirical distribution, $\mathbb{P}_{N_{\nu}} = \frac{1}{N_{v}}\mathop{\sum}_{i=1}^{N_{\nu}}\delta\left(\mbf{X}_{i}\right)$, used to approximate $P \in \mathcal{P}$, the distribution hypothesized to generate the finite population, $U_{\nu}$.

We follow the notational convention of \citet{Ghosal00convergencerates} and define the associated expectation functionals with respect to these empirical distributions by $\mathbb{P}^{\pi}_{N_{\nu}}f = \frac{1}{N_{\nu}}\mathop{\sum}_{i=1}^{N_{\nu}}\frac{\delta_{\nu i}}{\pi_{\nu i}}f\left(\mbf{X}_{i}\right)$.  Similarly, $\mathbb{P}_{N_{\nu}}f = \frac{1}{N_{\nu}}\mathop{\sum}_{i=1}^{N_{\nu}}f\left(\mbf{X}_{i}\right)$.  Lastly, we use the associated centered empirical processes, $\mathbb{G}^{\pi}_{N_{\nu}} = \sqrt{N_{\nu}}\left(\mathbb{P}^{\pi}_{N_{\nu}}-P_{0}\right)$ and $\mathbb{G}_{N_{\nu}} = \sqrt{N_{\nu}}\left(\mathbb{P}_{N_{\nu}}-P_{0}\right)$.

The sampling-weighted, (average) pseudo-Hellinger distance between distributions, $P_{1}, P_{2} \in \mathcal{P}$, $d^{\pi,2}_{N_{\nu}}\left(p_{1},p_{2}\right) = \frac{1}{N_{\nu}}\mathop{\sum}_{i=1}^{N_{\nu}}\frac{\delta_{\nu i}}{\pi_{\nu i}}d^{2}\left(p_{1}(\mathbf{X}_{i}),p_{2}(\mathbf{X}_{i})\right)$, where $d\left(p_{1},p_{2}\right) = \left[\mathop{\int}\left(\sqrt{p_{1}}-\sqrt{p_{2}}\right)^{2}d\mu\right]^{\frac{1}{2}}$ (for dominating measure, $\mu$).
We need this empirical average distance metric because the observed (sample) data drawn from the finite population under $P_{\nu}$ are no longer independent.
The associated non-sampling Hellinger distance is specified with, $d^{2}_{N_{\nu}}\left(p_{1},p_{2}\right) = \frac{1}{N_{\nu}}\mathop{\sum}_{i=1}^{N_{\nu}}d^{2}\left(p_{1}(\mathbf{X}_{i}),p_{2}(\mathbf{X}_{i})\right)$.

\subsection{Main result}\label{main:results}
We proceed to construct associated conditions and a theorem that contain our main result on the consistency of the pseudo-posterior distribution under a broader class of informative sampling designs at the true generating distribution, $P_{0}$.  This approach follows the main in-probability convergence result of \citet{2015arXiv150707050S} which extends \citet{ghosal2007} by adding new conditions that restrict the distribution of the informative sampling design.
Instead of the standard asymptotic factorization condition, we provide two alternative conditions which allow for residual dependence between sampling units:

Suppose we have a  sequence, $\xi_{N_{\nu}} \downarrow 0$ and $N_{\nu}\xi^{2}_{N_{\nu}}\uparrow\infty$  and $n_{\nu}\xi^{2}_{N_{\nu}}\uparrow\infty$ as $\nu\in\mathbb{Z}^{+}~\uparrow\infty$ and any constant, $C >0$,

\begin{description}
\item[(A1)\label{existtests}] (Local entropy condition - Size of model)
        \begin{equation*}
        \mathop{\sup}_{\xi > \xi_{N_{\nu}}}\log N\left(\xi/36,\{P\in\mathcal{P}_{N_{\nu}}: d_{N_{\nu}}\left(P,P_{0}\right) < \xi\},d_{N_{\nu}}\right) \leq N_{\nu} \xi_{N_{\nu}}^{2},
        \end{equation*}
\item[(A2)\label{sizespace}] (Size of space)
        \begin{equation*}
        \displaystyle\Pi\left(\mathcal{P}\backslash\mathcal{P}_{N_{\nu}}\right) \leq \exp\left(-N_{\nu}\xi^{2}_{N_{\nu}}\left(2(1+2C)\right)\right)
        \end{equation*}
\item[(A3)\label{priortruth}] (Prior mass covering the truth)
        \begin{equation*}
        \displaystyle\Pi\left(P: -P_{0}\log\frac{p}{p_{0}}\leq \xi^{2}_{N_{\nu}}\cap P_{0}\left[\log\frac{p}{p_{0}}\right]^{2}\leq \xi^{2}_{N_{\nu}} \right) \geq \exp\left(-N_{\nu}\xi^{2}_{N_{\nu}}C\right)
        \end{equation*}
\item[(A4)\label{bounded}] (Non-zero Inclusion Probabilities)
        \begin{equation*}
        \displaystyle\mathop{\sup}_{\nu}\left[\frac{1}{\displaystyle\mathop{\min}_{i \in U_{\nu}}\vert\pi_{\nu i}\vert}\right] \leq \gamma, \text{  with $P_{0}-$probability $1$.}
        \end{equation*}

\item[(A5.1)\label{deprestrict}] (Growth of dependence is restricted)\\
For every $U_{\nu}$ there exists a binary partition $\{S_{\nu 1}, S_{\nu 2}\}$ of the set of all pairs $S_{\nu}= \{\{i,j\}: i\ne j \in U_{\nu}\}$ such that
        \begin{equation*}
         \displaystyle\mathop{\limsup}_{\nu\uparrow\infty} \left\vert S_{\nu 1} \right\vert = \order{N_{\nu}},
        \end{equation*}
        and
 	 \begin{equation*}
        \displaystyle\mathop{\limsup}_{\nu\uparrow\infty} \mathop{\max}_{i,j \in S_{\nu 2}}\left\vert\frac{\pi_{\nu ij}}{\pi_{\nu i}\pi_{\nu j}} - 1\right\vert = \order{N_{\nu}^{-1}}, \text{  with $P_{0}-$probability $1$}
        \end{equation*}
        such that for some constants, $C_{4},C_{5} > 0$ and for $N_{\nu}$ sufficiently large,
         \begin{equation*}
         \left\vert S_{\nu 1} \right\vert \le C_{4} N_{\nu},
        \end{equation*}
        and
        \begin{equation*}
        \displaystyle N_{\nu}\mathop{\sup}_{\nu}\mathop{\max}_{i,j \in S_{\nu 2}}\left\vert\frac{\pi_{\nu ij}}{\pi_{\nu i}\pi_{\nu j}} - 1\right\vert \leq C_{5},
        \end{equation*}
	
\item[(A5.2)\label{depblock}] (Dependence restricted to countable blocks of bounded size)\\	
For every $U_{\nu}$ there exists a partition $\{B_1,\dots, B_{D_{\nu}}\}$ of $U_{\nu}$ with $D_{\nu} \le N_{\nu}$,
$\mathop{\lim}_{\nu\uparrow\infty} D_{\nu} = \order{N_{\nu}}$, and the maximum size of each subset is bounded:
        \begin{equation*}
        1 \le \displaystyle\mathop{\sup}_{\nu} \displaystyle\mathop{\max}_{d \in 1,\dots, D_{\nu}}\left\vert B_d\right\vert \leq C_{4},
        \end{equation*}
Such that the set of all pairs $S_{\nu}= \{\{i,j\}: i\ne j \in U_{\nu}\}$ can be partitioned into
 $S_{\nu 1} = \left\{\{i,j\}: i\ne j \in  B_{d}, d \in \{1, \ldots, D_{\nu} \}\right\}$ and\\
$S_{\nu 2} = \left\{\{i,j\}: i \in B_d \cap j \notin B_{d}, d \in \{1, \ldots, D_{\nu}\}\right\}$ with
	 \begin{equation*}
        \displaystyle\mathop{\limsup}_{\nu\uparrow\infty} \mathop{\max}_{i,j \in S_{\nu_2}}\left\vert\frac{\pi_{\nu ij}}{\pi_{\nu i}\pi_{\nu j}} - 1\right\vert = \order{N_{\nu}^{-1}}, \text{  with $P_{0}-$probability $1$}
        \end{equation*}
        such that for some constant, $C_{5} > 0$,
        \begin{equation*}
        \displaystyle N_{\nu}\mathop{\sup}_{\nu}\mathop{\max}_{i,j \in S_{\nu_2}}\left\vert\frac{\pi_{\nu ij}}{\pi_{\nu i}\pi_{\nu j}} - 1\right\vert \leq C_{5}, \text{  for $N_{\nu}$ sufficiently large.}
        \end{equation*}

\item[(A6)\label{fraction}] (Constant Sampling fraction)
        For some constant, $f \in(0,1)$, that we term the ``sampling fraction",
        \begin{equation*}
        \mathop{\limsup}_{\nu}\displaystyle\biggl\vert\frac{n_{\nu}}{N_{\nu}} - f\biggl\vert = \order{1}, \text{  with $P_{0}-$probability $1$.}
        \end{equation*}

\end{description}
The first three conditions are the same as  \citet{ghosal2007}. They restrict the growth rate of the model space (e.g., of parameters) and require prior mass to be placed on an interval containing the true value.
Condition~\nameref{bounded} requires the sampling design to assign a positive probability for inclusion of every unit in the population because the restriction bounds the sampling inclusion probabilities away from 0. Condition~\nameref{fraction} ensures that the observed sample size, $n_{\nu}$, limits to $\infty$ along with the size of the partially-observed finite population, $N_{\nu}$, such that the variation of information about the population expressed in realized samples is controlled.

\citet{2015arXiv150707050S} rely on asymptotic factorization for all pairwise inclusion probabilities. Their (A.5) condition is a conservative approach to establish a finite upper bound for the un-normalized posterior mass assigned to those models, $P$, at some minimum distance from the truth, $P_0$.   They require all terms, a set of size $\order{N_{\nu}^{2}}$, to factorize with the maximum deviation term shrinking at a rate of $\order{N_{\nu}^{-1}}$, since there are $N^2$ terms divided by $N$ (inherited from an empirical process).

Although their condition guarantees the $L_1$ contraction result, it defines an overly narrow class of sampling designs under which this guaranteed result holds.  As discussed in the introduction, multistage household survey designs are not members of this allowed class because the within household dependency does not attenuate for a set of pairs of size $\order{N_{\nu}}$.  We replace their (A5) with \nameref{deprestrict}, which allows up to $\order{N_{\nu}}$ pairwise terms to not factor, such that there remains a residual dependence.   We show in the Appendix that the contraction result may, nevertheless, be guaranteed under this condition as each of the non-factoring terms has an $\order{1}$ bound.   The implication of our condition is that we have constructed a wider class of sampling designs that includes those from \citet{2015arXiv150707050S}, in addition to the multistage cluster designs for fixed cluster sizes.

Our condition \nameref{depblock} is a special case of \nameref{deprestrict} specified for cluster designs where the number of units per cluster is bounded by a constant, which encompasses the multistage NSDUH household design from which we draw our application data set.   We walk from \nameref{deprestrict} to \nameref{depblock} by constructing $S_{\nu 1}$ through a collection of clusters $(B_{\nu 1},…,B_{\nu D_{\nu}})$,  where the size $|B_{\nu d}|$ is bounded from above.  Sampling dependence within each cluster $B_{\nu d}$ is unrestricted, while dependence across clusters must asymptotically factor.

\begin{theorem}
\label{main}
Suppose conditions ~\nameref{existtests}-\nameref{fraction} hold.  Then for sets $\mathcal{P}_{N_{\nu}}\subset\mathcal{P}$, constants, $K >0$, and $M$ sufficiently large,
\begin{align}\label{limit}
&\mathbb{E}_{P_{0},P_{\nu}}\Pi^{\pi}\left(P:d^{\pi}_{N_{\nu}}\left(P,P_{0}\right) \geq M\xi_{N_{\nu}} \vert \mbf{X}_{1}\delta_{\nu 1},\ldots,\mbf{X}_{N_{\nu}}\delta_{\nu N_{\nu}}\right) \leq\nonumber\\
&\frac{16\gamma^{2}\left[\gamma \mathbf{C_{2}} +C_{3}\right]}{\left(Kf + 1 - 2\gamma\right)^{2}N_{\nu}\xi_{N_{\nu}}^{2}} + 5\gamma\exp\left(-\frac{K n_{\nu}\xi_{N_{\nu}}^{2}}{2\gamma}\right),
\end{align}
which tends to $0$ as $\left(n_{\nu}, N_{\nu}\right)\uparrow\infty$.
\end{theorem}

\begin{proof}
The proof follows exactly that in \citet{2015arXiv150707050S} where we bound the numerator (from above) and the denominator (from below) of the expectation with respect to the joint distribution of population generation and the taking of a sample of the pseudo-posterior mass placed on the set of models, $P$, at some minimum pseudo-Hellinger distance from $P_{0}$. We reformulate one of the enabling lemmas of \citet{2015arXiv150707050S}, which we present in an Appendix, where the reliance on (their) condition (A5) requiring asymptotic factoring of pairwise unit inclusion probabilities is here replaced by condition  \nameref{deprestrict} that allows for non-factorization of a subset of pairwise inclusion probabilities.
\end{proof}

As noted in \citet{2015arXiv150707050S}, the rate of convergence is decreased for a sampling distribution, $P_{\nu}$, that expresses a large variance in unit pairwise inclusion probabilities such that $\gamma$ will be relatively larger. Samples drawn under a design that expresses a large variability in the first order sampling weights will express more dispersion in their information relative to a simple random sample of the underlying finite population. We construct $C_3 = C_5 + 1$ and $C_2 = C_4 + 1$.
Under the more restrictive condition (A5) of \citet{2015arXiv150707050S}, our constant $C_4 = 0$ and thus $C_2 = 1$.

\section{Simulation Examples}\label{sec:sims}
We construct a population model to address our inferential interest of a binary outcome $y$ with a linear predictor $\mu$.
\begin{comment}Even though our motivating example, the association between past year smoking and past year major depressive episode (MDE) are two binary measures, for ease of demonstration we use a binary outcome and continuous predictor for our simulation study.
\end{comment}
\begin{equation} \label{pop_like}
y_{i} \mid \mu_{i} \ind Bern \left(F^{-1}_l(\mu_{i}) \right),~ i = 1,\ldots,N
\end{equation}
where $F^{-1}_l$ is the quantile function (inverse cumulative function) for the logistic distribution.
We let $\mu$ depend on two predictors $x_1$ and $x_2$. The variable $x_1$ represents the observed information available for analysis, whereas $x_2$ represents information available for sampling, which is either ignored or not available for analysis. The $x_1$ and $x_2$ distributions are $\mathcal{N}(0,1)$ and $\mathcal{E}(r =1/5)$ with rate $r$, where $\mathcal{N}(\cdot)$ and $\mathcal{E}(\cdot)$ represent normal and exponential distributions, respectively. The size measure used for sample selection is  $\tilde{\bm{x}}_{2} = \bm{x}_{2} - \min (\bm{x}_{2}) + 1$.
\[
\bm{\mu} = -1.88 + 1.0 \bm{x}_{1}+ 0.5 \bm{x}_{2}
\]
where the intercept was chosen such that the median of $\mu$ is approximately 0, therefore the median of $F^{-1}_l(\bm{\mu})$ is approximately 0.5.

Even though the population response $y$ was simulated with $\mu = f(x_1,x_2)$, we estimate the marginal models at the population level for $\mu = f(x_1)$. This exclusion of $x_2$ is analogous to the situation in which an analyst does not have access to all the sample design information and ensures that our sampling design instantiates informativeness (where $y$ is correlated with the selection variable, $x_{2}$, that defines inclusion probabilities). In particular, we estimate the models under each of several sample design scenarios and compare the population fitted models, $\mu = f(x_1)$, to those from the samples.

We formulate the logarithm of the sampling-weighted pseudo-likelihood for estimating $(\bm{\mu},\lambda)$ from our observed data for the $ n\leq N$ sampled units,
\begin{align}\label{pseudo_like}
\log\left[\mathop{\prod}_{i=1}^{n} p\left(y_{i}\mid \mu_{i}\right)^{w^{\ast}_{i}}\right] &= \mathop{\sum}_{i=1}^{n}w^{\ast}_{i}\log p \left( y_{i}\mid \mu_{i}\right) \nonumber\\
&= \mathop{\sum}_{i=1}^{n} w^{\ast}_{i} y_{i} \log(\theta_i) + w^{\ast}_{i} (1- y_{i}) \log(1-\theta_i),
\end{align}
where $\theta_i =  F^{-1}_l(\mu_{i})$ and the sampling weights, $w^{\ast}_{i}$ are normalized such that the sum of the weights equals the sample size $\mathop{\sum}_{i=1}^{n}w^{\ast}_{i} = n$.

Finally, we estimate the joint posterior distribution using Equation~\ref{pseudo_like}, coupled with our prior distributions assignments, using the NUTS Hamiltonian Monte Carlo algorithm implemented in Stan \citep{stan:2015}.

\FloatBarrier
\subsection{Multistage Cluster Designs}
We begin by abstracting the five-stage, geographically-indexed NSDUH sampling design \citep{MRB:Sampling:2014} to a simpler, three stage design of \{area segment, household, individual\} that we use to draw samples from a synthetic population in a manner that still generalizes to the NSDUH (and similar multistage sampling designs where the number of last stage units does not grow with overall population size).
We simulate a population of N = 6000, with 200 primary sampling units (PSUs) each containing 10 households (HHs) which each contain 3 individuals with independent responses $y_i$.

For the simulation, the number of selected PSUs was varied $K \in \{10, 20, 40, 80, 160\}$, the number of HHs within each PSU was fixed at 5, and the number of selected individuals within each HH was 1. Each setting was repeated $M = 200$ times. Details for the selection at each stage follows:
\begin{enumerate}
	\item For each PSU indexed by $k$, an aggregate size measure $X_{2,k} = \sum_{ij} x_{2,ij|k}$ was created summing over all individuals $i$ and HHs $j$ in PSU $k$. PSUs are then selected proportional to this size measure based on Brewer's PPS algorithm \citep{BrewerPPS}.
	\item Once PSUs are selected, for each HH within the selected PSUs indexed by $j$, an aggregate size measure $X_{2,j|k} = \sum_{i} x_{2,i|jk}$ was created summing over all individuals $i$ within each HH in the selected PSUs. HHs are selected independently across PSUs. Within each PSU, HHs are selected systematically with equal probability by first sorting on $X_{2,j|k}$ and then selecting a random starting point.
	\item Within each selected HH, a single person is selected with probability proportional to size $x_{2,i|jk}$.
\end{enumerate}

The nested structure of the sampling induces asymptotic independence between PSU's. Within PSUs, the systematic sampling of HHs creates a block of non-attenuating dependence between households. Likewise, the sampling of only one person within each HH creates a joint dependence $\pi_{ii'|jk} = 0$ between individuals within the same HH. Therefore, non-factorization of the second order inclusions remains within each PSU (see Figure \ref{fig:3stagefactor}).
Figure \ref{fig:bin2pred} compares the bias and mean square error (MSE) for estimation with equal weights (black) and inverse probability weights (blue). As expected, the sampling weights remove bias and lead to convergence, since the non-factoring pairwise inclusion probablities are of $\order{N}$.

\begin{figure}
\centering
\includegraphics[width = 0.45\textwidth,
		page = 1,clip = true, trim = 0.7in 0in 0.7in 0in]{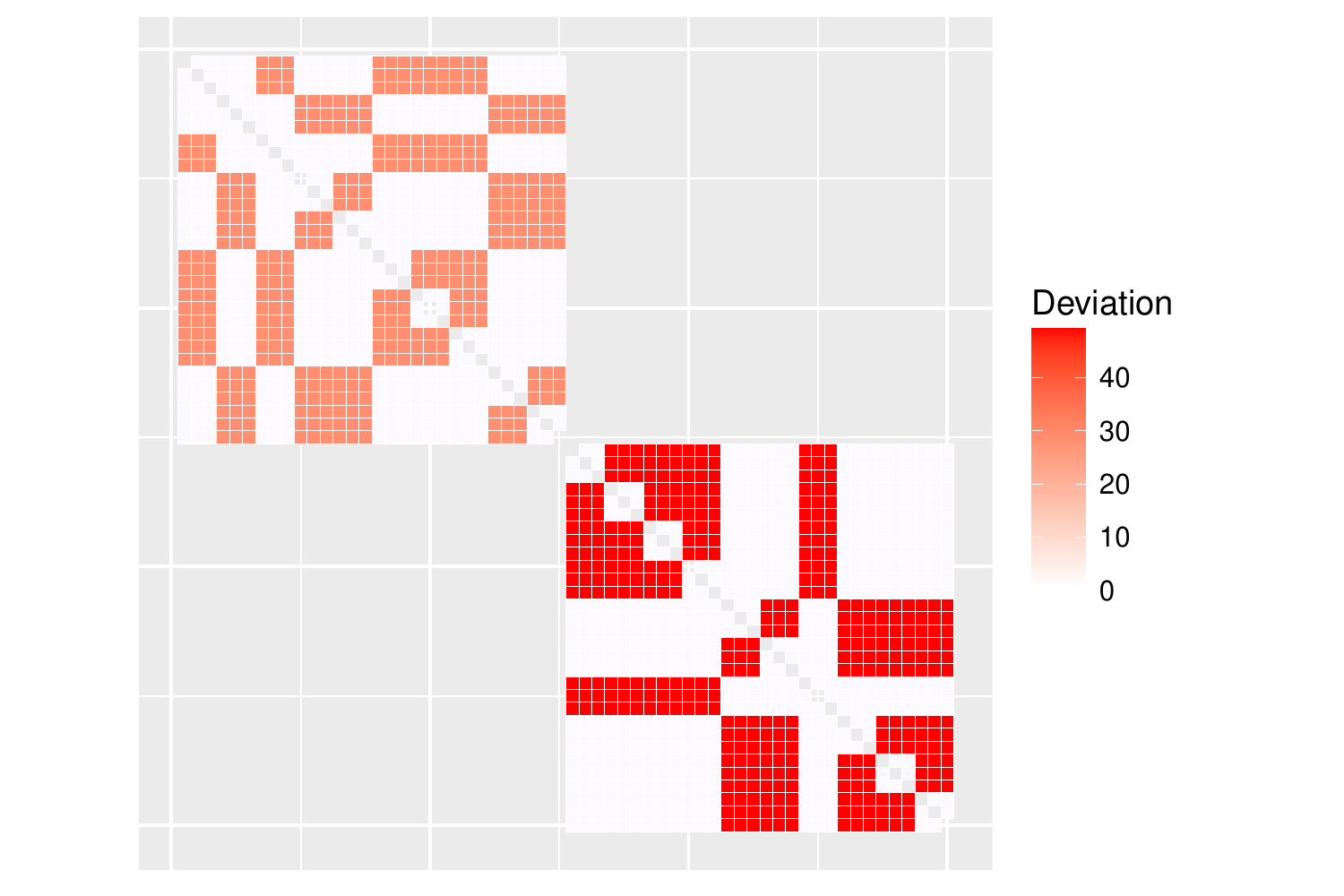}
\includegraphics[width = 0.45\textwidth,
		page = 1,clip = true, trim = 0.7in 0in 0.7in 0in]{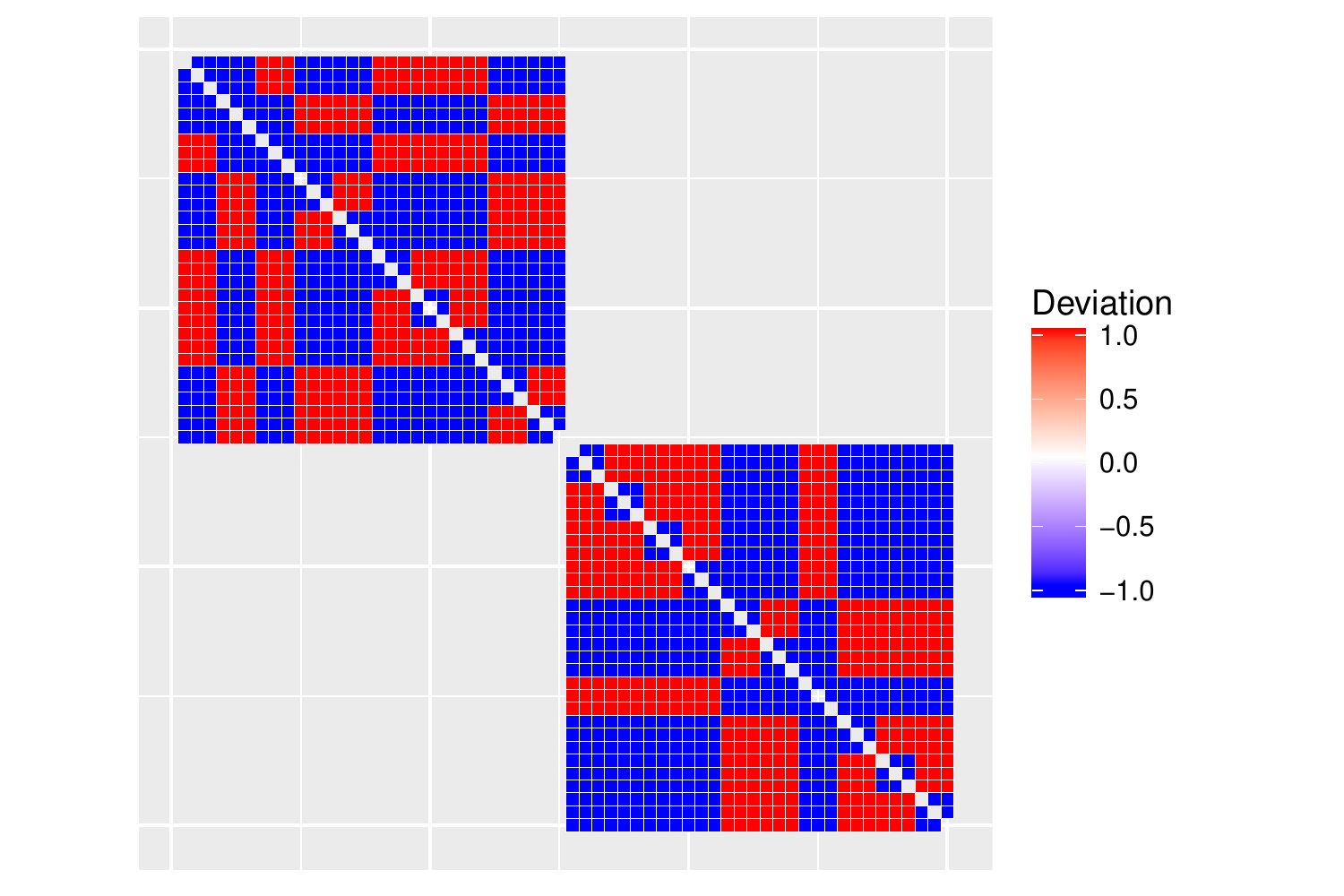}
\caption{Matrix $\{i,j\}$ of deviations from factorization $\left(\pi_{ij}/(\pi_{i} \pi_{j}) - 1\right)$ for two PSUs (out of a population of 200) from the three stage sample design. Each PSU contains 10 HHs, which each contain 3 persons. Magnitude (left) and sign (right) of deviations. Empty cells correspond to $0$ deviation (factorization).}
\label{fig:3stagefactor}
\end{figure}

\begin{figure}
\centering
\includegraphics[width = 0.95\textwidth,
		page = 1,clip = true, trim = 0in 0in 0in 0.in]{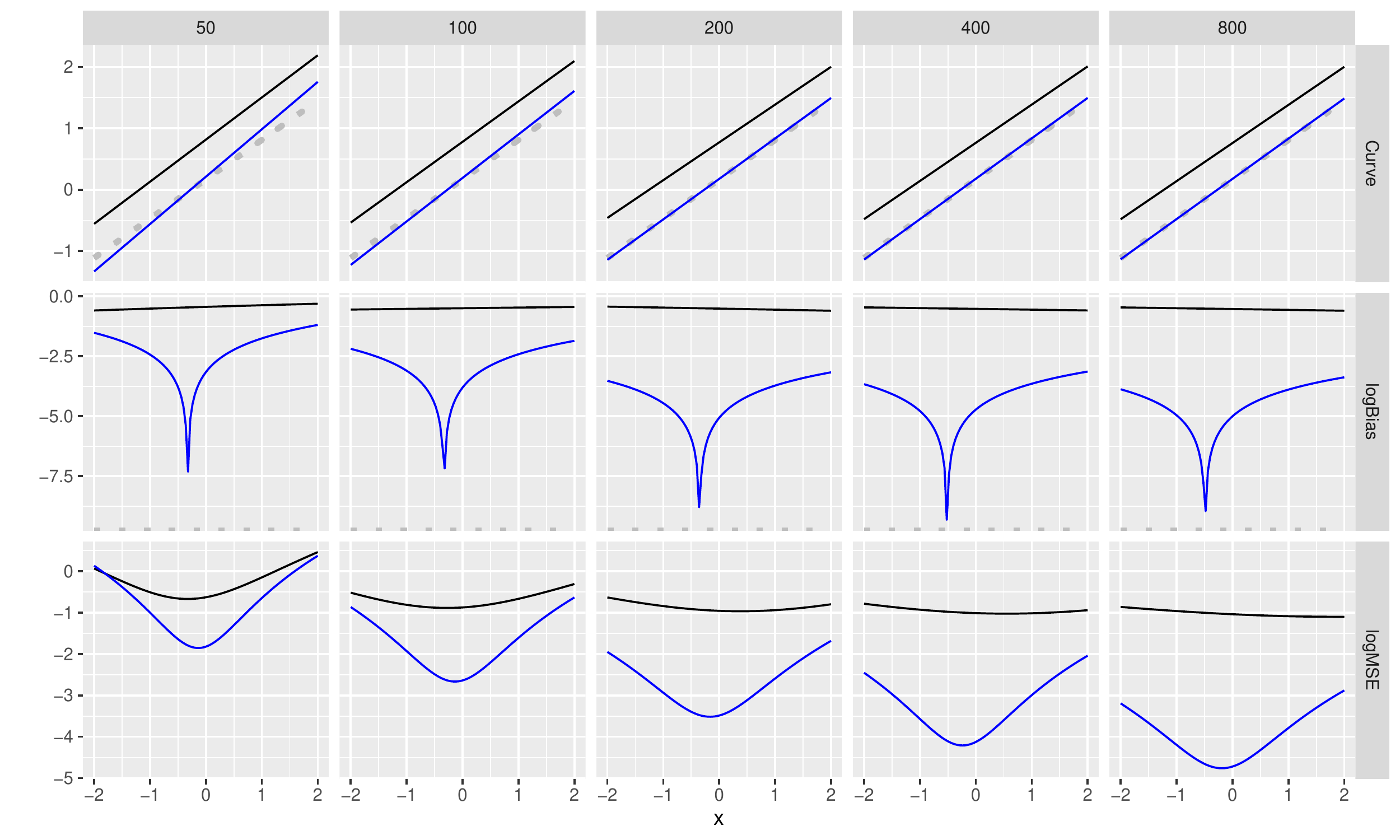}
\caption{The marginal estimate of $\mu = f(x_1)$ under a linear model and three-stage sampling design. Compares the (true) population curve (broken grey) to the whole sample with equal weights (black), and inverse probability weights (blue). Top to bottom: estimated curve, log of absolute bias, log of mean square error. Left to right: doubling of sample size (50 to 800). }
\label{fig:bin2pred}
\end{figure}
\FloatBarrier

\subsection{Dependent Sampling of First Stage Units}
We now use the same population response model and distributions for $y$, $x_1$, and $x_2$ but consider the case of single stage sampling designs where the sample size is half the population (i.e. a partition of size $N/2$). In particular, we construct a design with second order dependence that grows $\order{N^{2}}$ and demonstrate that estimates for this design fail to converge. However, with slight modifications, the design can be altered into $\order{N}$ dependence and does demonstrate convergence, as predicted by the theory.

One simple way to create an informative design is to use the size measure $\tilde{\bm{x}}_{2}$ to sort the population.  Partition the population $U$ into a ``high'' ($U_1$) group with the top $N/2$ and a ``low''  ($U_2$) group with the bottom $N/2$.
This partition rule leads to an outcome space with only two possible samples of size $N/2$: $U_1$ and $U_2$. For simplicity, assume an equal probability of selection of $1/2$. Then it follows that $\pi_{i} = 1/2$, for all $i \in 1,\ldots, N$, and $\pi_{ij} = 1/2$ if $i \ne j \in U_k$, for $k = 1,2$ and 0 otherwise. In fact, all joint inclusions, from orders 2 to $N/2$, are 1/2 if all members indexed are in the same partition and 0 otherwise. These second and higher order inclusion probabilities do not factor with increasing population size $N$. Thus, the number of pairwise inclusions probabilities that do not factor ($\pi_{ij} \ne 1/4$) grows at rate $\order{N^{2}}$, violating condition~\nameref{deprestrict}.

Alternatively, we could embed the partitioning procedures within strata, where the strata are created according to rank order, have a fixed size, and the number of strata grow with population size $N$. For example grouping every 50 units into a strata, then partitioning within each. Such a modification is relatively minor, but leads to factorization for all but $\order{N}$ pairwise inclusion probabilities. This can be visualized as the diagonal blocks in the full pairwise inclusion matrix (see Figure \ref{fig:factorstrata}).

\begin{figure}
\centering
\includegraphics[width = 0.95\textwidth,
		page = 1,clip = true, trim = 0in 0in 0in 0in]{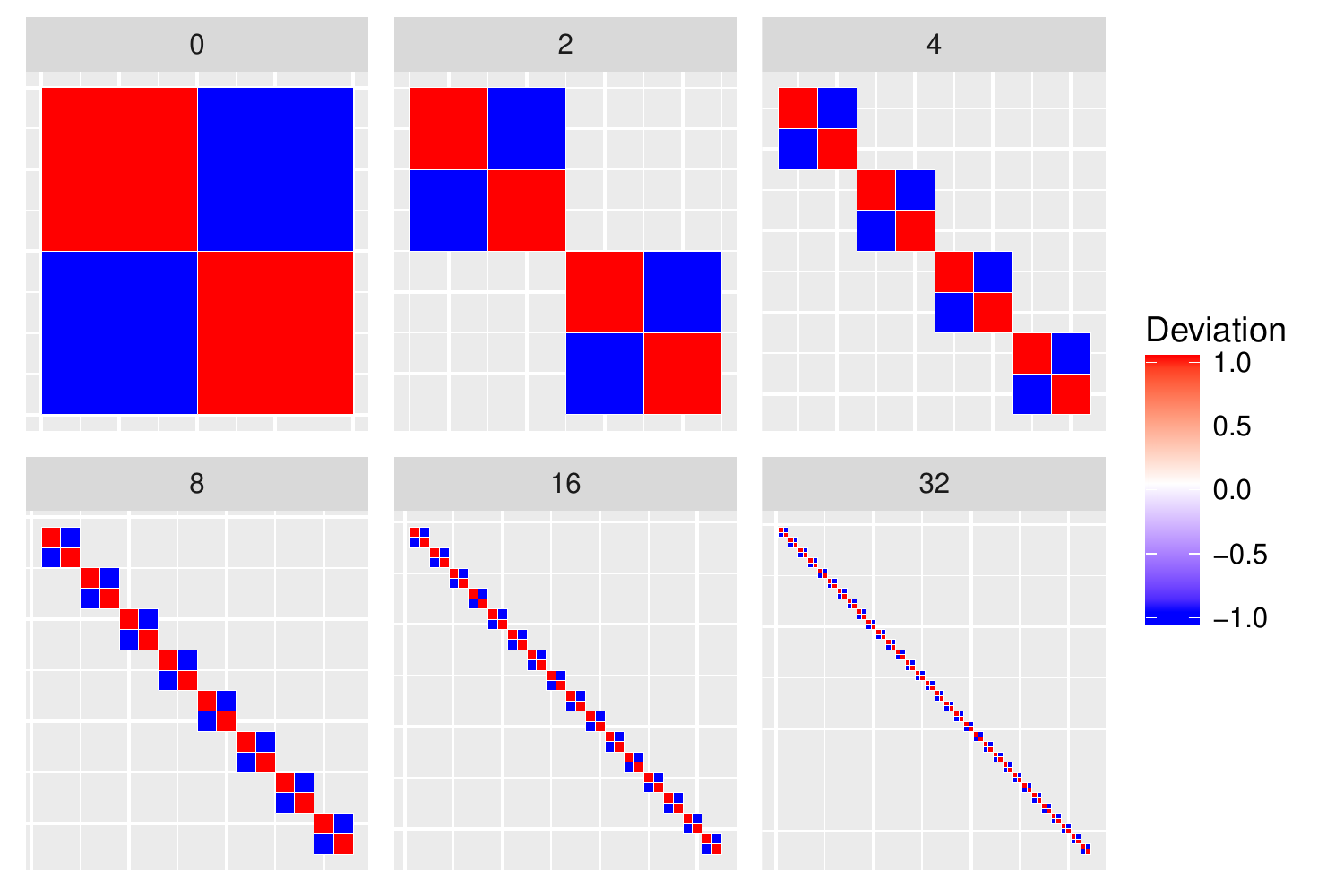}
\caption{Matrix $\{i,j\}$ of deviations from factorization $\left(\pi_{ij}/(\pi_{i} \pi_{j}) - 1\right)$ for an equal probability dyadic partition design by number of strata (0 to 32). Empty cells correspond to $0$ deviation (factorization).}
\label{fig:factorstrata}
\end{figure}

For each $N \in \{100, 200, 400, 800, 1600\}$, we generate a single population and compare the relative convergence of the original dyadic partitions and the stratified versions. Figure \ref{fig:rankpar} compares the bias and mean square error (MSE) of the two partitions (red and blue) compared to the average of 100 samples from the stratified version (black). It's clear that as the population size (and sample size) grows, the bias of the two partitions does not go away (the variability is due to a single realization of the population at each size), while the overall bias and MSE of the stratified version clearly decreases with increasing N, consistent with the theory.

\begin{figure}
\centering
\includegraphics[width = 0.95\textwidth,
		page = 1,clip = true, trim = 0in 0in 0in 0.in]{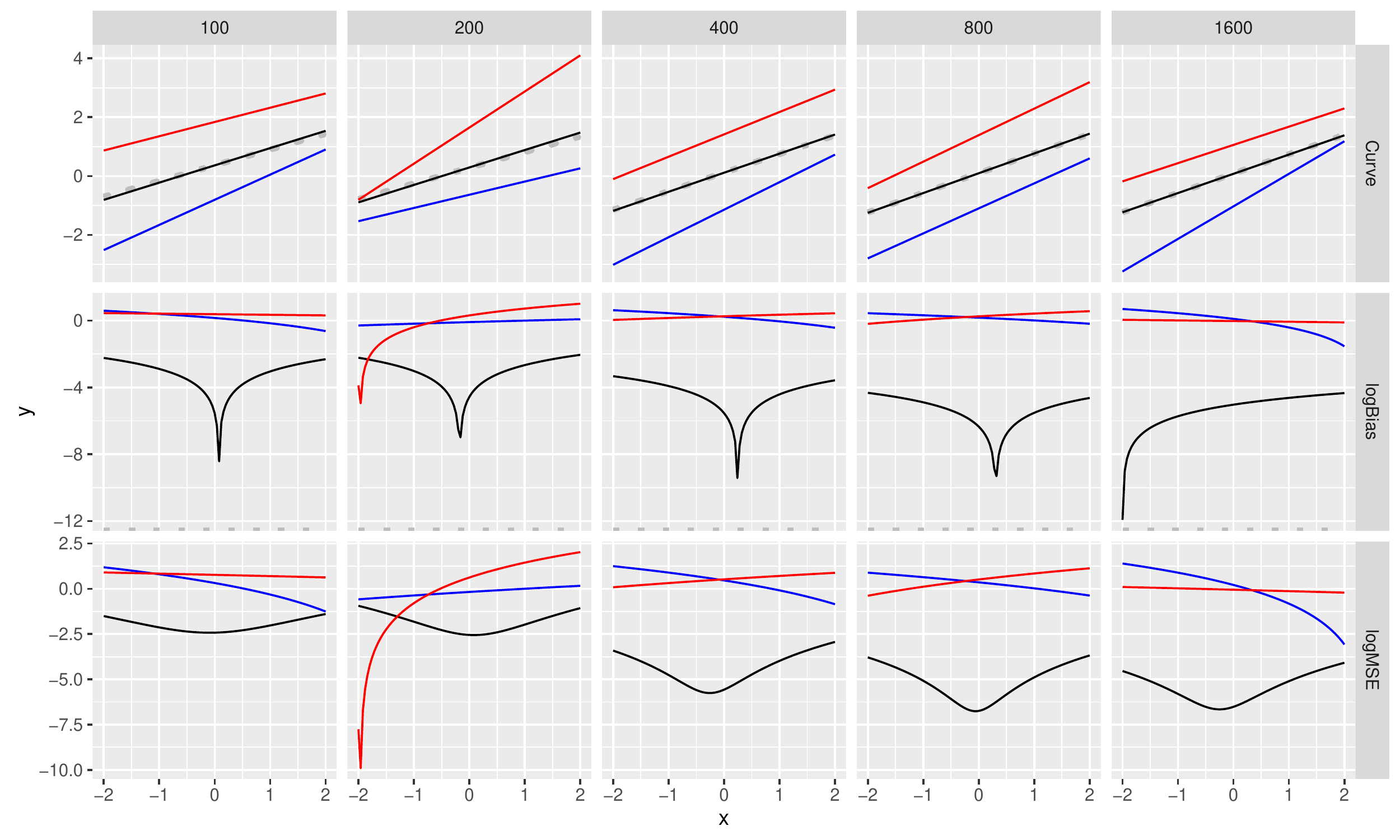}
\caption{The marginal estimate of $\mu = f(x_1)$ under a linear model and one-stage sampling design with dyadic partitions. Compares the (true) population curve (broken grey) to binary partition (red and blue) and the stratified partition sample (black) - one stratum per 50 individuals, each divided into a binary partition, repeated 100 times. Top to bottom: estimated curve, log of absolute bias, log of mean square error. Left to right: doubling of population size (100 to 1600).}
\label{fig:rankpar}
\end{figure}

\FloatBarrier

\section{Application to the NSDUH}
\label{sec:NSDUH}
%develop a bit more - relate design features to sim results and assumptions A5.2

A simple logistic model of current (past month) smoking status by past year major depressive episode (MDE) was fit via the survey weighted psuedo-posterior as described in section \ref{sec:sims} using both equal and probability-based analysis weights for adults from the 2014 NSDUH public use data set (Figure \ref{fig:NSDUH}).
It is reasonable to assume that equal weights lead to higher estimates of smoking, as young adults are more likely to smoke and are over-sampled. Based on the theoretical results and the simulation study presented in this paper, we have justification that the probability-based weights have removed this bias and provide consistent estimation. The large number of strata and the asymptotically independent first stage of selection creates factorization for all but $\order{N}$ pairwise inclusion probabilities, even though the clustering and the sorting of units before selection may be informative.

\begin{figure}
\centering
\includegraphics[width = 0.75\textwidth,
		page = 1,clip = true, trim = 0in 0in 0in 0.in]{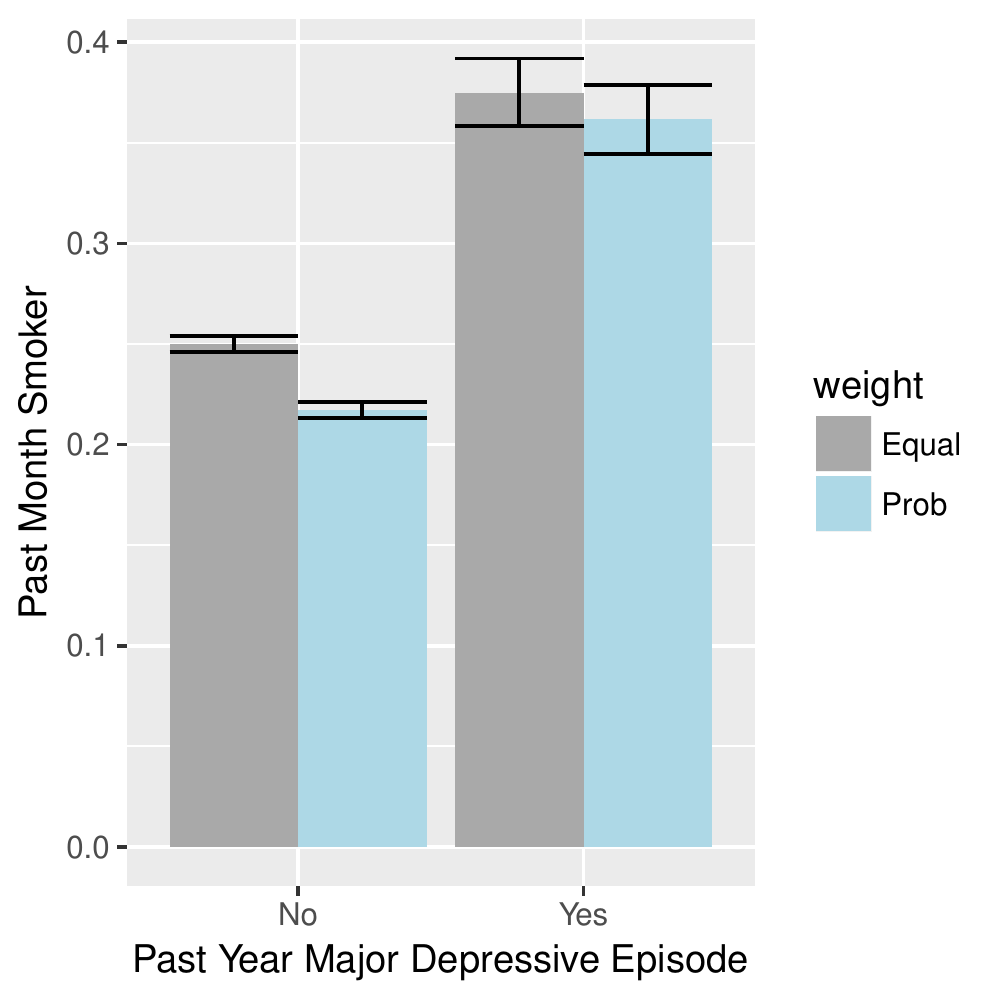}
\caption{Posterior estimates (mean and 95\% intervals) for adults in the US who smoked cigarettes in the past month by past year major depressive episode, using equal weights (grey) and probability based analysis weights (blue) based on the 2014 National Survey on Drug Use and Health.}
\label{fig:NSDUH}
\end{figure}
\FloatBarrier

\section{Conclusions}
\label{sec:conc}

This work is motivated by the discrepancy between the theory available to justify consistent estimation for survey sample designs and the practice of estimation for complex, multistage cluster designs such as the NSDUH. Previous requirements for approximate or asymptotic factorization of joint sampling probabilities exclude such designs, leaving the practitioner unable to fully justify their use. We have presented an alternative requirement that allows for unrestricted sampling dependence to persist asymptotically rather than to attenuate. For example, dependence between units within a cluster is unrestricted provided that the cluster size is bounded and dependence between clusters attenuates. This dependence can be positive (joint selection) or negative (mutual exclusion). Results are further demonstrated via a simulation study of a simplified NSDUH design. %joint inclusion of 0, mutually exclusive - allowable here - contrast to pair%

Additional simulations expand our understanding of the impact of sorting. While the direct application of these methods can lead to dependence among all units (effectively one cluster of infinite size), embedding these features within stratified or clustered designs can be justified (for subsequent estimation using marginal sampling weights) by our main results and performs well in simulation and in practice. For example, geographic units sorted along a gradient can now be fully justified for the NSDUH, because the sampling along this gradient occurs independently across a large number of strata.

With this work, the use of the sample weighted pseudo-posterior \citep{2015arXiv150707050S} is now available to a much wider variety of survey programs. We note that while establishing consistency is essential, understanding other properties of pseudo-posteriors such as posterior intervals, still requires more research. Furthermore, so called ``Fully Bayesian'' methods, which avoid a plug-in estimator for the sampling weights by jointly modelling the outcome and the sample selection process, are also being researched \citep{2017arXiv171000019N}. The theory uses the stricter conditions for asymptotic factorization of the sample design and could be generalized by using conditions for the sample design that are similar to those presented in this work.

\appendix

\section{Enabling Lemmas}
%mention first lemma and relaxation in other paper
\begin{lemma}\label{numerator}
Suppose conditions~\nameref{existtests} and ~\nameref{bounded} hold.  Then for every $\xi > \xi_{N_{\nu}}$, a constant, $K>0$, and any constant, $\delta > 0$,
\begin{align}
\mathbb{E}_{P_{0},P_{\nu}}\left[\mathop{\int}_{P\in\mathcal{P}\backslash\mathcal{P}_{N_{\nu}}}\mathop{\prod}_{i=1}^{N_{\nu}}
\frac{p^{\pi}}{p_{0}^{\pi}}\left(\mbf{X}_{i}\delta_{\nu i}\right)d\Pi\left(P\right)\left(1-\phi_{n_{\nu}}\right)\right] \leq \Pi\left(\mathcal{P}\backslash\mathcal{P}_{N_{\nu}}\right)& \label{outside}\\
\mathbb{E}_{P_{0},P_{\nu}}\left[\mathop{\int}_{P\in\mathcal{P}_{N_{\nu}}:d^{\pi}_{N_{\nu}}\left(P,P_{0}\right)> \delta\xi}\mathop{\prod}_{i=1}^{N_{\nu}}
\frac{p^{\pi}}{p_{0}^{\pi}}\left(\mbf{X}_{i}\delta_{\nu i}\right)d\Pi\left(P\right)\left(1-\phi_{n_{\nu}}\right)\right] &\leq \nonumber \\
2\gamma\exp\left(\frac{-K n_{\nu}\delta^{2}\xi^{2}}{\gamma}\right).&\label{inside}
\end{align}
\end{lemma}

\begin{proof}
See  \citet{2015arXiv150707050S}, \citet{2016arXiv160607488S}, and \citet{2017pair} for details and modifications.
\end{proof}
%move on to second lemma
\begin{lemma}\label{denominator}
For every $\xi > 0$ and measure $\Pi$ on the set,
\begin{equation*}
B = \left\{P:-P_{0}\log\left(\frac{p}{p_{0}}\right) \leq \xi^2, P_{0}\left(\log\frac{p}{p_{0}}\right)^{2} \leq \xi^{2}\right\}
\end{equation*}
under the conditions ~\nameref{sizespace}, ~\nameref{priortruth}, ~\nameref{bounded}, ~\nameref{deprestrict}, we have for every $C > 0 $,  $C_{2} = C_{4}+1$, $C_{3} = C_{5}+1$, and $N_{\nu}$ sufficiently large,
\begin{equation}\label{denomresult}
\mbox{Pr}\left\{\mathop{\int}_{P\in\mathcal{P}}\displaystyle\mathop{\prod}_{i=1}^{N_{\nu}}\frac{p^{\pi}}{p_{0}^{\pi}}
\left(\mbf{X}_{i}\delta_{\nu i}\right)d\Pi\left(P\right)\leq \exp\left[-(1+C)N_{\nu}\xi^{2}\right]\right\}
\leq \frac{\gamma \mathbf{C_{2}}+C_{3}}{C^{2} N_{\nu}\xi^{2}},
\end{equation}
where the above probability is taken with the respect to $P_{0}$ and the sampling generating distribution, $P_{\nu}$, jointly.
\end{lemma}

\begin{proof}\label{AppDenominator}
The proof follows that of \citet{2015arXiv150707050S} by bounding the probability expression on left-hand size of Equation~\ref{denomresult} with,
\begin{align}
&\mbox{Pr}\left\{\mathbb{G}^{\pi}_{N_{\nu}}\mathop{\int}_{P\in\mathcal{P}}\log\frac{p}{p_{0}}
d\Pi\left(P\right)\leq -\sqrt{N_{\nu}}\xi^{2}C\right\}\nonumber\\
&\leq\frac{\displaystyle\mathop{\int}_{P\in\mathcal{P}}\left[\mathbb{E}_{P_{0},P_{\nu}}
\left(\mathbb{G}^{\pi}_{N_{\nu}}\log\frac{p}{p_{0}}\right)^{2}\right]d\Pi\left(P\right)}
{N_{\nu}\xi^{4}C^{2}}\label{chebyshev:e2},
\end{align}
where we have used Chebyshev to achieve the right-hand bound of Equation~\ref{chebyshev:e2}.  We now proceed to further bound the numerator in the right-hand side of Equation~\ref{chebyshev:e2}. \citet{2015arXiv150707050S} and \citet{2017pair} establish the following:
\begin{equation}\label{gbound}
\mathbb{E}_{P_{0},P_{\nu}}\left[\mathbb{G}^{\pi}_{N_{\nu}}\log\frac{p}{p_{0}}\right]^{2} \leq N_{\nu}\mathbb{E}_{P_{0},P_{\nu}}\left[\left(\mathbb{P}^{\pi}_{N_{\nu}} - \mathbb{P}_{N_{\nu}}\right)\log\frac{p}{p_{0}}\right]^{2} +~ \xi^{2}
\end{equation}
We proceed to further simplify the bound in the first term on the right in Equation~\ref{gbound}:
\begin{subequations}
\begin{align}
&N_{\nu}\mathbb{E}_{P_{0},P_{\nu}}\left[\left(\mathbb{P}^{\pi}_{N_{\nu}} - \mathbb{P}_{N_{\nu}}\right)\log\frac{p}{p_{0}}\right]^{2}\nonumber\\
&= \displaystyle\frac{1}{N_{\nu}}\mathop{\sum}_{i=j\in U_{\nu}}\mathbb{E}_{P_{0}}\left[\left(\frac{1}
{\pi_{\nu i}} - 1 \right)\left(\log\frac{p}{p_{0}}\left(\mbf{X}_{i}\right)\right)^{2}\right]\nonumber\\
&+ \frac{1}{N_{\nu}}\mathop{\sum}_{i=j \in U_{\nu}}\mathbb{E}_{P_{0}}\left[\left(\frac{\pi_{\nu ij}}{\pi_{\nu i}\pi_{\nu j}}-1\right)\log\frac{p}{p_{0}}\left(\mbf{X}_{i}\right)\log\frac{p}{p_{0}}\left(\mbf{X}_{j}\right)\right]\label{firstbro}\\
&\le \displaystyle\frac{1}{N_{\nu}}\mathop{\sum}_{i=j\in U_{\nu}}\mathbb{E}_{P_{0}}\left[\left(\frac{1}
{\pi_{\nu i}} - 1 \right)\left(\log\frac{p}{p_{0}}\left(\mbf{X}_{i}\right)\right)^{2}\right]\nonumber\\
&+ \frac{1}{N_{\nu}}\mathop{\sum}_{i=j \in U_{\nu}} \left|\mathbb{E}_{P_{0}}\left[\left(\frac{\pi_{\nu ij}}{\pi_{\nu i}\pi_{\nu j}}-1\right)\log\frac{p}{p_{0}}\left(\mbf{X}_{i}\right)\log\frac{p}{p_{0}}\left(\mbf{X}_{j}\right)\right]\right|\label{secondbro}\\
&= \displaystyle\frac{1}{N_{\nu}}\mathop{\sum}_{i=j\in U_{\nu}}\mathbb{E}_{P_{0}}\left[\left(\frac{1}
{\pi_{\nu i}} - 1 \right)\left(\log\frac{p}{p_{0}}\left(\mbf{X}_{i}\right)\right)^{2}\right]\nonumber\\
&+ \frac{1}{N_{\nu}}\mathop{\sum}_{i\ne j\in S_{\nu 1}} \left|\mathbb{E}_{P_{0}}\left[\left(\frac{\pi_{\nu ij}}{\pi_{\nu i}\pi_{\nu j}}-1\right)\log\frac{p}{p_{0}}\left(\mbf{X}_{i}\right)\log\frac{p}{p_{0}}\left(\mbf{X}_{j}\right)\right]\right|\nonumber\\
&+ \frac{1}{N_{\nu}}\mathop{\sum}_{i \ne j\in S_{\nu 2}}\left|\mathbb{E}_{P_{0}}\left[\left(\frac{\pi_{\nu ij}}{\pi_{\nu i}\pi_{\nu j}}-1\right)\log\frac{p}{p_{0}}\left(\mbf{X}_{i}\right)\log\frac{p}{p_{0}}\left(\mbf{X}_{j}\right)\right]\right|\label{thirdbro}\\
&\leq \left(\gamma - 1 \right)\displaystyle\frac{1}{N_{\nu}}\mathop{\sum}_{i=j\in U_{\nu}}\mathbb{E}_{P_{0}}\left[\left(\log\frac{p}{p_{0}}\left(\mbf{X}_{i}\right)\right)^{2}\right]\nonumber\\
&+ \max\{1,\gamma - 1\}\frac{1}{N_{\nu}}\mathop{\sum}_{S_{\nu 1}}\left|\mathbb{E}_{P_{0}}\left[\log\frac{p}{p_{0}}\left(\mbf{X}_{i}\right)\log\frac{p}{p_{0}}\left(\mbf{X}_{j}\right)\right]\right|\nonumber\\
&+ C_{5} N_{\nu}^{-1}\frac{1}{N_{\nu}}\mathop{\sum}_{S_{\nu 2}}\left|\mathbb{E}_{P_{0}}\left[\log\frac{p}{p_{0}}\left(\mbf{X}_{i}\right)\log\frac{p}{p_{0}}\left(\mbf{X}_{j}\right)\right]\right|\label{fourthbro}\\
&\leq\left(\gamma (1 +C_{4}) +  C_{5}\right)\xi^{2}\label{fifthbro},
\end{align}
\end{subequations}
for sufficiently large $N_{\nu}$.
The first equality (\ref{firstbro}) is derived from the quadratic expansion and subsequent expectation of the inclusion indicators $\delta_{\nu i},\delta_{\nu j}$ with respect to the conditional distribution of $P_{\nu}$ given and follows \citet{2015arXiv150707050S} and \citet{2017pair}. The next inequality (\ref{secondbro}) is needed because the pairwise terms could be negative, so the sum is bounded by the sum of the absolute value.  The two pairwise terms are equivalently partitioned by $S_{\nu 1}$ and  $S_{\nu 2}$ in the next equality (\ref{thirdbro}). Condition ~\nameref{bounded} implies the following bounds:
\begin{equation}
-1 \le \left(\frac{\pi_{\nu ij}}{\pi_{\nu i}\pi_{\nu j}}-1\right) \le
\left (\frac{1}{\pi_{\nu i}} - 1 \right) \le  (\gamma - 1)
\end{equation}
since $0 \le \pi_{\nu ij} \le \min\{\pi_{\nu i},\pi_{\nu j}\}$, which are used in \ref{fourthbro}. The size bounds from condition ~\nameref{deprestrict} and  the definition of the space $B$ provide the remaining bounds (in \ref{fifthbro}).

We may now bound the expectation on the right-hand size of Equation~\ref{chebyshev:e2},
\begin{multline}
\mathbb{E}_{P_{0},P_{\nu}}\left[\mathbb{G}^{\pi}_{N_{\nu}}\log\frac{p}{p_{0}}\right]^{2} \leq \left(\gamma (1 +C_{4}) +  C_{5}\right)\xi^{2}+ \xi^2 \\ \leq \left(\gamma (1 +C_{4}) +  C_{5} + 1\right)\xi^{2} = (\gamma C_{2} +  C_{3})\xi^2,
\end{multline}
for $N_{\nu}$ sufficiently large, where we set $C_{2} := C_{4}+1$ and $C_{3} := C_{5}+1$.  This concludes the proof.
\end{proof}

\bibliography{refs_june2018}
\bibliographystyle{agsm}

\end{document}